\documentclass[11pt,twoside]{article}
\usepackage[pdftex]{graphicx}
\usepackage{amsmath}
\usepackage{amssymb}
\usepackage{cite}
\usepackage{amstext,amsthm}
\usepackage{graphicx}
\usepackage{epstopdf}
\usepackage{color}
\usepackage{bm}
\usepackage{appendix}
\usepackage[T1]{fontenc}
\usepackage{bbm}
\usepackage{latexsym}

\usepackage{float}
\usepackage{verbatim}
\usepackage{lipsum}
\usepackage{braket}

\newcommand{\del}[0]{\partial}

\newtheorem{proposition}{Proposition}

\begin{document}
\newcommand{\pst}{\hspace*{1.5em}}

\newcommand{\rigmark}{\em Journal of Russian Laser Research}
\newcommand{\lemark}{\em Volume 30, Number 5, 2009}

\newcommand{\be}{\begin{equation}}
\newcommand{\ee}{\end{equation}}
\newcommand{\ds}{\displaystyle}
\newcommand{\bea}{\begin{eqnarray}}
\newcommand{\eea}{\end{eqnarray}}
\newcommand{\ba}{\begin{array}}
\newcommand{\ea}{\end{array}}
\newcommand{\arcsinh}{\mathop{\rm arcsinh}\nolimits}
\newcommand{\arctanh}{\mathop{\rm arctanh}\nolimits}
\newcommand{\bc}{\begin{center}}
\newcommand{\ec}{\end{center}}

\thispagestyle{plain}

\label{sh}


\begin{center} {\Large \bf
\begin{tabular}{c}
EFFICIENT TRAINABILITY OF LINEAR OPTICAL MODULES 
\\[-1mm]
IN QUANTUM OPTICAL NEURAL NETWORKS
\end{tabular}
 } \end{center}

\bigskip

\bigskip

\begin{center} {\bf
T.J. Volkoff
}\end{center}

\medskip

\begin{center}
{\it
Theoretical Division, Los Alamos National Laboratory,\\
Los Alamos, NM, USA.

\smallskip


$^*$Corresponding author e-mail:~~~volkoff@lanl.gov\\
}\end{center}

\begin{abstract}\noindent
The existence of ``barren plateau landscapes'' for generic discrete variable quantum neural networks, which obstructs efficient gradient-based optimization of cost functions defined by global measurements, would be surprising in the case of generic linear optical modules in quantum optical neural networks due to the tunability of the intensity of continuous variable states and the relevant unitary group having exponentially smaller dimension. We demonstrate that coherent light in $m$ modes can be generically compiled efficiently if the total intensity scales sublinearly with $m$, and extend this result to cost functions based on homodyne, heterodyne, or photon detection measurement statistics, and to noisy cost functions in the presence of attenuation. We further demonstrate efficient trainability of $m$ mode linear optical quantum circuits for variational mean field energy estimation of positive quadratic Hamiltonians for input states that do not have energy exponentially vanishing with $m$.
\end{abstract}

\medskip


\section{Introduction}\label{sc:intro}
Advances in low-loss integrated photonics suggest that quantum optical neural networks (QONN) provide a feasible hardware framework for implementing near-term continuous variable variational quantum algorithms \cite{engreview,engproc,peruzzo2014VQE,yama,PhysRevResearch.1.033063}. For example, variational protocols such as compiling a unitary operation \cite{engunsamp} and classification \cite{engdeep} have been demonstrated within the QONN framework.  However, for discrete variable variational quantum algorithms, the existence of barren plateau landscapes (BPL) for cost functions based on expectation values and generic parameterized circuits indicates that gradient descent optimization of cost functions is asymptotically inefficient (with respect to the number of qubits) \cite{mcclean2018barren,mcclean2016}. Because quantum neural networks are often constructed from randomly initialized quantum circuits, this fact constitutes a ``no-go'' theorem for generic scalability of quantum machine learning protocols.  For continuous variable (CV) variational quantum algorithms, the question remains open whether randomly initialized CV quantum circuits can be efficiently trained by applying a gradient-based optimization algorithm to a relevant cost function. In the present work, we answer this question for a large class of CV variational quantum algorithms that implement generic parameterized linear optical circuits to process coherent state inputs.

We begin with a technical definition of BPL behavior. A continuously differentiable cost function $C:\Theta \rightarrow \mathbb{R}_{+}$ on a probability space of parameters $(\Theta,p(\theta)d\theta)$ exhibits BPL on a compact subset $\mathcal{A} \subset \Theta$ if for every $\epsilon >0$, there exists $0<b<1$ such that \begin{equation}P_{\mathcal{A}}\left( \Big\vert{\del C \over \del \theta_{j}}\Big\vert \ge \epsilon \right) = \mathcal{O}(b^{m}),\label{eqn:bpldef}\end{equation} where $P_{\mathcal{A}}$ is the probability measure on $\mathcal{A}$ induced from $p$, and $m$ is the number of modes in the quantum computation (i.e., the number of qudits in the discrete variable case or number of photonic modes in the continuous variable case) \cite{vc}. In variational quantum algorithms, the probability space $\Theta$ is the parameter space of a
random parameterized quantum circuit (RPQC) $U(\theta)$, and the cost function $C$ is a function of $U(\theta)$, the input
state $\ket{\psi_{0}}$, and an observable $H$ (see (\ref{eqn:cf}) below). In the discrete variable case \cite{mcclean2018barren}, the quantum circuit $U(\theta)$ is constructed from a random network of parameterized and unparameterized gates acting on a small number of qubits (this structure is sometimes called a \textit{quantum neural network}). In the present work, $U(\theta)$ is constructed from a random network of parameterized and unparameterized beamsplitters and phase shifters (i.e., a \textit{quantum optical neural network (QONN)}) as discussed in Sec. \ref{sec:gen}. To analyze the BPL phenomenon for generic variational quantum algorithms, it is required that expectations with respect to the probability density $p(\theta)d\theta$ match the Haar measure on the appropriate compact Lie group (e.g., the unitary group $\mathcal{U}(2^{m})$ for the case of $m$-qubits, or the orthogonal group $O(2m)$ for the case of $m$ linear optical modes). This requires circuits $U(\theta)$ of sufficiently large depth \cite{brandaohoro}.

To prove that a cost function $C$ exhibits BPL, we will often use Chebyshev's inequality in the form
\begin{align}
P\left( \Big\vert {\del C\over \del \theta_{j}} \Big\vert >\epsilon\right) &\le {E\left( \vert\del_{\theta_{j}} C\vert^{2} \right) \over \epsilon^{2}} \text{ or } \nonumber \\
P\left( \Big\vert {\del C\over \del \theta_{j}} \Big\vert >\epsilon\right) &\le {E\left( \vert\del_{\theta_{j}} C\vert \right) \over \epsilon}
\label{eqn:cheby}
\end{align}
to get an upper bound on the tail probability of the gradient of $C$. The computation of the expectation on the right hand side therefore comprises the main technical part of the analyses.

Conversely, a lower bound on the expected magnitude of the gradient of $C(\theta)$ that vanishes slower than exponentially with $m$ indicates the possibility, in principle, of efficient  gradient descent optimization of $C(\theta)$ \cite{cfdbp}.  When $C(\theta)$ has the form
\begin{equation}
C(\theta)=\langle \psi(\theta) \vert H \vert \psi(\theta) \rangle
\label{eqn:cf}
\end{equation}
where $\ket{\psi(\theta)}=U(\theta)\ket{\psi_{0}}$ for input state $\ket{\psi_{0}}$, $U(\theta)$ is a RPQC, and $H$ is a positive operator, the absence of BPL for $C(\theta)$ is necessary for efficient trainability of the RPQC $U(\theta)$ by gradient descent.

In this work, we provide sufficient conditions on the intensity of input coherent states that allow to avoid BPL for certain cost functions defined by (\ref{eqn:cf}) with $U(\theta)$ a linear optical RPQC. The Hamiltonians $H$ that we consider correspond to the tasks of compiling a linear optical unitary operation and variational mean field energy estimation of quadratic bosonic Hamiltonians. We find that the BPL is absent for these
cost functions if the intensity of the input state scales sublinearly with the number of CV modes and the same intensity
does not vanish exponentially with the number of CV modes (Secs. \ref{sec:toy}-\ref{sec:mfvqe}). Compared to the results for generic discrete variable RPQC in \cite{mcclean2018barren}, trainability of generic linear optical CV RPQC is possible due to two factors. Firstly, for $m$ linear optical modes, note that the RPQC $U(\theta)$ acts on an input coherent state on $m$-modes by a phase space rotation of the coherent state vector. Therefore, a generic linear optical RPQC is associated with a Haar distributed $2m \times 2m$ orthogonal matrix. Because $O(2m)$ is isomorphic to the unitary group $\mathcal{U}(m)$, it follows that the relevant unitary group for describing generic linear optical CV RPQC is $\mathcal{U}(m)$, of dimension $m^{2}$. This fact contrasts with the case of generic RPQC on $m$ qubit registers, because such an RPQC is associated with a Haar distributed $2^{m}\times 2^{m}$ unitary matrix, i.e., the relevant unitary group is $\mathcal{U}(2^{m})$ of dimension $2^{2m}$, exponentially larger than the case of linear optical RPQC. Secondly, because the distinguishability of CV coherent states depends on intensity \cite{tvtd}, which is preserved by linear optical transformations and can be tuned relative to the circuit size, the existence of a generic BPL phenomenon for cost functions that depend on input coherent states and linear optical RPQC $U(\theta)$ is expected to depend on the intensity. By contrast, in the discrete variable case, the orbit of $\mathcal{U}(2^{m})$ is dense in the $m$ qubit pure state space, from which it follows that the output of Haar distributed discrete variable RPQC is independent of the input state.

Both for simplicity and relevance to CV quantum communication tasks, we restrict our analysis of trainability to linear optical RPQCs acting on  coherent state inputs.  A coherent state $\ket{\vec{u}}$ of $m$ modes  is defined as the unique Gaussian state that has equal and minimal fluctuations of the canonical observables $R=(q_{1},p_{1},\ldots , q_{m},p_{m})$ and has mean vector $\vec{u}=\langle R \rangle_{\ket{\vec{u}}} \in \mathbb{R}^{2m}$ \cite{holevobook}. In Secs. \ref{sec:toy} and \ref{sec:gen}, the cost functions we consider take the form of $C(\theta)$ in (\ref{eqn:cf}) with $H=\mathbb{I}-\ket{\vec{u}}\bra{\vec{u}}$ and input state $\ket{\psi_{0}}=\ket{\vec{u}}$. This choice of $H$ is actually not as restrictive as it appears, since we extend the analysis to cost functions based on homodyne, heterodyne, and photon number detection measurements. Recently, a verification protocol for an NP complete satisfiability problem was demonstrated using linear optical circuits, photon number detection, and $m$-mode coherent inputs with local intensity scaling as $O(m^{-1/4})$ \cite{kerenidis}. In that proposal, the ``power of unentanglement'' \cite{aaronsonpower} (which allows for sublinear proof size in quantum Merlin-Arthur protocols for NP verification) lies in the possibility of tuning the intensity of the coherent state that encodes a satisfiability proof. Our BPL analysis following from Proposition \ref{prop:one} indicate that this sublinear scaling of the total intensity is within the intensity range that allows to avoid BPL for generic linear optical RPQCs and cost functions based on photon detection measurements.    
Therefore, our results suggest that variational linear optical versions of quantum Merlin-Arthur verification protocols for NP are not expected to exhibit the BPL phenomenon. Such variational verifiers may be useful when Arthur does not know how Merlin's proof is encoded in the local phases of a coherent state.

In all sections of this work, $H$ in (\ref{eqn:cf}) is a positive operator on $m$ quantum harmonic oscillators of frequency $\omega=1$, which we call modes. The RPQCs $U(\theta)$ that we consider are given by energy-conserving linear optical unitary operators on $m$ modes. Specifically, every  unitary operator $U$ in the RPQC acts on the row vector of canonical operators via $U^{*}RU=RT$ for a matrix $T\in O(2m)$, the orthogonal group (which is isomorphic to $\mathcal{U}(m)$). It follows that $U\ket{\vec{u}}=\ket{\vec{u}T}$ for any coherent state $\ket{\vec{u}}$. We use an asterisk to denote the adjoint of an operator on Hilbert space.

\section{Intensity dependence of BPL: basic example\label{sec:toy}}

To illustrate the intensity dependence of the BPL phenomenon in a simple linear optical circuit, we first consider a cost function of the form (3) with an initial state given by an $m$th tensor product of a single mode coherent state, i.e., $\ket{\psi_{0}}=\left( \ket{  u_{1}+iu_{2}/ \sqrt{2}} \right)^{\otimes m}=: \ket{\vec{u}}$ with $\vec{u}=(u_{1},u_{2})^{\oplus m} \in \mathbb{R}^{2m}$ (this notation is chosen so that the vector of canonical operators $R=(q_{1},p_{1},\ldots,q_{m},p_{m})$  has expectation $\langle \vec{u}\vert R \vert \vec{u}\rangle = \vec{u}$, where $q_{j}=(a_{j}+a_{j}^{*})/ \sqrt{2}, p_{j}=(-ia_{j}+ia_{j}^{*})/ \sqrt{2}$, $j=1,\ldots m$). We consider a simple RPQC given by a local phase-shifter $U(\theta)=e^{-i\sum_{j=1}^{m}\theta_{j} a^{*}_{j}a_{j}}$ that preserves the local intensity, and we use the Hamiltonian $H=\mathbb{I}-\ket{\vec{u}}\bra{\vec{u}}$ to define the cost function $C(\theta)$. With this framework, the expectation of $\del_{\theta_{k}}C$ with respect to the uniform measure on $[-\pi , \pi]^{\times m}$ iz zero due to the parity of the integrand. However, the expectation of the magnitude $\vert\del_{\theta_{1}}C\vert$ results in the expression (use $\alpha:= {u_{1}+iu_{2}\over \sqrt{2}}$ so that $\ket{\psi_{0}}=\ket{\alpha}^{\otimes m}$)
\begin{align}
E\left( \vert\del_{\theta_{1}}C\vert \right) &= {1\over (2\pi)^{m}}\int d\vec{\theta} \, \big\vert\del_{\theta_{1}}\left( 1- \vert \bra{\alpha}^{\otimes m}U(\theta)\ket{\alpha}^{\otimes m} \vert^{2}\right)  \big\vert \nonumber \\
&= {1\over (2\pi)^{m}}\int d\vec{\theta} \, \Big\vert\del_{\theta_{1}}\left( 1-e^{-2m\vert \alpha \vert^{2}}\prod_{j=1}^{m}e^{2\vert \alpha\vert^{2}\cos \theta_{j}}\right)  \Big\vert \nonumber \\
&= 2e^{-2m\vert \alpha\vert^{2}}\vert\alpha\vert^{2} I_{0}(2\vert \alpha\vert^{2})^{m-1}\int_{-\pi}^{\pi} {d\theta_{1}\over 2\pi} \, \vert \sin \theta_{1}\vert e^{2\vert \alpha\vert^{2}\cos \theta_{1}}\nonumber \\
&= {2\over \pi}e^{-m(u_{1}^{2}+u_{2}^{2})}I_{0}(u_{1}^{2}+u_{2}^{2})^{m-1} \sinh (u_{1}^{2}+u_{2}^{2})
\label{eqn:modbess}
\end{align}
where $I_{0}(x)= {1\over 2\pi}\int_{-\pi}^{\pi}du \, e^{x\cos u}$ is a modified Bessel function of the first kind. Using the large argument asymptotics  $I_{0}(x)\sim {e^{x}\over \sqrt{2\pi x}}$ \cite{as}, one finds that, as a function of $u_{1}^{2}+u_{2}^{2}$,
\begin{equation}
E\left( \vert\del_{\theta_{1}}C\vert \right) \sim {2\over \pi (2\pi (u_{1}^{2}+u_{2}^{2}))^{m-1}}
\label{eqn:ltlt}
\end{equation}
which indicates that the derivative of the cost function vanishes exponentially with the number of modes in the limit of a large local intensity, i.e., $C(\theta)$ exhibits BPL. For example, if the local intensity grows linearly with $m$, i.e., $u_{1}^{2}+u_{2}^{2} = \mathcal{O}(m)$, then the BPL definition (\ref{eqn:bpldef}) is clearly satisfied due to the inequality (\ref{eqn:cheby}) and the fact that the expectation (\ref{eqn:ltlt}) goes to zero faster than exponentially. Note that the local intensity of $\ket{\psi_{0}}$ is $\langle {u_{1}+iu_{2}\over \sqrt{2}} \vert a^{*}a \vert  {u_{1}+iu_{2}\over \sqrt{2}}\rangle = {u_{1}^{2}+u^{2}_{2}\over 2}$ and the total intensity is ${\Vert \vec{u}\Vert^{2}\over 2}$.

On the other hand, by choosing $u_{1}$ and $u_{2}$ so that the local intensity scales sufficiently slowly with the mode number $m$, it is possible for the cost function (\ref{eqn:modbess}) to avoid the BPL phenomenon. Then one uses the small argument asymptotic $\lim_{x\rightarrow 0}I_{0}(x)=1$ of the modified Bessel function. For example, taking $u_{1}^{2}+u_{2}^{2} \sim {\log m\over m^{\gamma}}$ with $0<\gamma \le 1$ in (\ref{eqn:modbess}), one obtains the following asymptotic scaling as a function of $m$
\begin{align}
E\left( \vert\del_{\theta_{1}}C\vert \right) &\sim {2\over \pi}e^{-m(u_{1}^{2}+u_{2}^{2})} \sinh (u_{1}^{2}+u_{2}^{2})\nonumber \\
&\sim  {2\over \pi}e^{-{\log m\over m^{\gamma -1 }}} \sinh ({\log m\over m^{\gamma}}) \nonumber \\
&=  {1\over \pi m^{m^{1-\gamma}-m^{-\gamma}}}-{1\over \pi m^{m^{1-\gamma}+m^{-\gamma}}}.
\label{eqn:thth}
\end{align}
By using the Chebyshev inequality (\ref{eqn:cheby}), one concludes that the scaling in (\ref{eqn:thth}) is sufficient for the absence of BPL as defined in (\ref{eqn:bpldef}) because for any $b>1$ and $0<\gamma\le 1$, $\lim_{m\rightarrow \infty}m^{m^{1-\gamma}}b^{-m}=0$, from which it follows that $E\left( \vert\del_{\theta_{1}}C\vert \right)$ does not decay exponentially. Tuning the trainability of an RPQC in this way by modulating the intensity has no analogue in the unencoded discrete variable setting with randomly initialized RPQC. However, an analogous tuning is possible if the parameterization of a discrete variable RPQC depends on the number of modes such a way that the RPQC approximates a phase space displacement in the limit of infinite modes \cite{vc,PhysRevA.73.052108}.

\section{Variational quantum compiling and Gaussian detection\label{sec:gen}}

We now consider a linear optical CV analogue of the analysis of BPL for generic discrete variable RPQCs in \cite{mcclean2018barren}. For the linear optical RPQC, we take an $L$ layer linear optical circuit $U(\theta)=U_{+}U_{-}$ for $U_{-}:= \prod_{\ell=1}^{k-1}U_{\ell}(\theta_{\ell})W_{\ell}$, $U_{+}:= \prod_{\ell=k}^{L}U_{\ell}(\theta_{\ell})W_{\ell}$. Layer $k$ is singled out because we will be considering training $\theta_{k}$, without loss of generality. The $W_{\ell}$ correspond to unparameterized linear optical unitaries, whereas for each $\ell$, there is a skew-symmetric $2m\times 2m$ matrix $D_{\ell}$ such that $U_{\ell}(\theta_{\ell})^{*}RU_{\ell}(\theta_{\ell})=Re^{\theta_{\ell}D_{\ell}}$. One can write $U_{\ell}(\theta_{\ell})=e^{-i\theta_{\ell}R\epsilon_{\ell}R^{T}}$ for a symmetric $2m\times 2m$ matrix $\epsilon_{\ell}$.
We assume that $U_{\pm}$ are associated with $O_{\pm}\in O(2m)$, respectively, and $O_{\pm}$ are independent and distributed according to Haar measure on $O(2m)$.  Due to the isomorphism $O(2m)\cong \mathcal{U}(m)$, a $t$-design on $O(2m)$ can be encoded in a $t$-design on $\log_{2}m$ qubits \cite{cleve}. 

Inserting initial coherent state $\ket{\psi_{0}}=\ket{\vec{u}}$, RPQC $U(\theta)$, and $H=\mathbb{I}-\ket{\vec{u}}\bra{\vec{u}}$ into (\ref{eqn:cf}), one evaluates the cost function to be
\begin{align}
C(\theta)&= 1-\vert \langle \vec{u} \vert U(\theta)\vert \vec{u} \rangle \vert^{2} \nonumber \\
&= 1-e^{-{1\over 2}\Vert \vec{u}(\mathbb{I}_{2m}-O_{-}O_{+} )\Vert^{2}}.
\label{eqn:cthet}
\end{align} 
This cost function corresponds to the task of variational compiling \cite{qaqc} of the identity operation on a set of isoenergetic coherent states.
To analyze the dependence of the BPL phenomenon on intensity, we take $E$ to be the energy input to the circuit, i.e., $\Vert \vec{u} \Vert = \sqrt{2E}$.
Taking column vectors $\vec{y}=O_{-}^{T}\vec{u}^{T}$, $\vec{b}=O_{+}\vec{u}^{T}$, it follows that
\begin{equation}
\partial_{\theta_{k}}C=-e^{-2E}\left( \vec{b}D_{k}\vec{y}\right) e^{\vec{b}^{T}\vec{y}}.
\label{eqn:simp}
\end{equation} 
From (\ref{eqn:cthet}), BPL phenomenon is expected when $E$ scales linearly with the number of modes due to the exponential concentration of the cost function at the constant value 1. To determine whether BPL is present for general intensity scaling, we calculate the expectation over $O_{\pm}$ of the square of (\ref{eqn:simp}) in the following proposition:

\begin{proposition}Let $\vec{u}$, $\vec{y}$, $\vec{b}$, $O_{\pm}$ and $E$ be as defined above. Let $\mathcal{B}=\lbrace \vec{e}_{1},\vec{e}_{2},\ldots ,\vec{e}_{2m}\rbrace$ be an orthonormal basis of $\mathbb{R}^{2m}$ with $\vec{e}_{1}=\vec{b}/\sqrt{2E}$ and let $\xi_{\textup{min(max)}}=\textup{min(max)}_{j}\Vert \vec{d}_{j}\Vert^{2}$ where $\vec{d}_{j}$ is a column of $D_{k}$. Then
\begin{align}
E_{O_{-}}\left( \partial_{\theta_{k}}C \right)&=0 \nonumber \\
E_{O_{+},O_{-}}\left( \left(\partial_{\theta_{k}}C\right)^{2} \right) &\in {e^{-4E}\Gamma(m)I_{m-1}(4E)\over 2m(2E)^{m-3}}[\xi_{\textup{min}},\xi_{\textup{max}}]
\label{eqn:prop1}
\end{align}
where $I_{m-1}(x)$ is the order $m-1$ modified Bessel function of the first kind.
\label{prop:one}
\end{proposition}

\begin{proof}
The right hand side of (\ref{eqn:prop1}) is a multiple of the closed interval $[\xi_{\textup{min}},\xi_{\textup{max}}] \subset \mathbb{R}_{+}$; therefore we seek upper and lower bounds to $E_{O_{+},O_{-}}\left( \left(\partial_{\theta_{k}}C\right)^{2} \right)$ of this form. The expectation over $O_{-}$ of (\ref{eqn:simp}) or its square is equivalent to taking an expectation of a function of $\vec{y}$ on the $(2m-1)$-sphere of radius $\sqrt{2E}$ with respect to the uniform probability density (the area is denoted $\Omega_{2m-1}(\sqrt{2E})$). From symmetry considerations, it is clear that $E_{O_{-}}$ of (\ref{eqn:simp}) is zero. The expectation of the square of (\ref{eqn:simp}) is given by
\begin{align}
{2Ee^{-4E}\over \Omega_{2m-1}(\sqrt{2E})}\int_{S^{2m-1}}\left(\vec{e}_{1}^{\,T}D_{k}\vec{y}\right)^{2}e^{4E\cos \varphi_{1}}
\label{eqn:ooo}
\end{align}
where the measure $\prod_{j=1}^{2m-1}d\varphi_{j}(\sqrt{2E})^{2m-1}\sin^{2m-2}\varphi_{1}\cdots \sin \varphi_{2m-2}$ is implicit in the integral (we used coordinates $\varphi_{1},\ldots ,\varphi_{2m-2}\in [0,\pi]$, $\varphi_{2m-1}\in [-\pi ,\pi]$). The symbol $\int_{S^{2m-1}}$ indicates integration with respect to this measure. Due to the angular integrations many terms in the integrand do not contribute, and (\ref{eqn:ooo}) can be written
\begin{align}
&{} {2Ee^{-4E}\over \Omega_{2m-1}(\sqrt{2E})}\int_{S^{2m-1}}\left(\sum_{j=1}^{2m}(D_{k})_{1,j}^{2} \vec{y}_{j}^{2}\right)e^{4E\cos \varphi_{1}} .
\label{eqn:inter}
\end{align}
The expectation over $O_{+}$ is now carried out on the expression (\ref{eqn:inter}) via 
\begin{align}
E_{O_{+}}((D_{k})_{1,j}^{2})&= E_{O_{+}}\left( \left(\vec{e}_{1}^{\,T}D_{k}\vec{e}_{j}\right)^{2} \right) \nonumber \\
&= {1\over 2E}E_{O_{+}}\left( \left(\vec{b}^{\,T}D_{k}\vec{e}_{j}\right)^{2} \right)\nonumber \\
&= {1\over 2E}E_{O_{+}}\left( \left(\vec{u}O_{+}^{T}D_{k}\vec{e}_{j}\right)^{2} \right)\nonumber \\
&={\Vert \vec{u}\Vert^{2}\over 4Em} \sum_{r=1}^{2m}(D_{k})_{r.j}^{2}\nonumber \\
&= {1\over 2m}\Vert \vec{d}_{j}\Vert^{2}
\end{align}

The formula (\ref{eqn:inter}) then simplifies to
\begin{align}
&{}{2Ee^{-4E}\over 2m\Omega_{2m-1}(\sqrt{2E})}\int_{S^{2m-1}}\left(\sum_{j=1}^{2m}\Vert \vec{d}_{j}\Vert^{2} \vec{y}_{j}^{2}\right)e^{4E\cos \varphi_{1}} \nonumber \\
&\in \Omega_{2m-2}(1)(2E)^{2}e^{-4E}{\Gamma(m)\over 2\pi^{m}}{f(E)\over 2m}[\xi_{\text{min}},\xi_{\text{max}}]
\end{align}
where $f(E):=\int_{0}^{\pi}d\varphi_{1}e^{4E\cos \varphi_{1}}\sin^{2m-2}\varphi_{1}$. The interval in (\ref{eqn:prop1}) is obtained by noting that \begin{align}
I_{m-1}(4E) &= {2^{-m}(4E)^{m-1}\over \sqrt{\pi}\Gamma(m-{1\over 2})}2f(E) \nonumber \\
&= {2^{-m}(4E)^{m-1}\over 2\pi^{m}}2\Omega_{2m-2}(1)f(E).
\end{align}
\end{proof}

The left end of the interval in Proposition \ref{prop:one} provides a lower bound on the expectation. 
We show in Appendix \ref{sec:app1} that if the intensity $E$ scales linearly (or vanishes exponentially) with $m$, the uniform asymptotics (respectively, small argument asymptotics) of the modified Bessel function imply the existence of BPL. There it is also shown that (\ref{eqn:prop1}) goes to zero sub-exponentially for total intensity scaling as $E\sim am^{r}$ with $r<1$, $a>0$. These observations indicate two intensity scaling transitions from trainability to BPL: one when $E$ is increased from vanishing exponentially to sublinear scaling and another when $E$ is decreased from linear scaling to sublinear scaling. This is the main result of the present work.

The conditions that allow to avoid BPL for the cost function (\ref{eqn:cthet}) also apply to cost functions based on the outcome of a photon number detection measurement, arising from, e.g., a coherent boson sampling protocol. In particular, a cost function (\ref{eqn:cf}) with $H=\mathbb{I}-\otimes_{j=1}^{m}\ket{n_{j}}\bra{n_{j}}$ and $\sum_{j=1}^{m}n_{j}=N$ and $\ket{\psi_{0}}=\ket{\vec{u}}$ is minimized when the RPQC $U(\theta)$ acts on $\ket{\vec{u}}$ to produce a coherent state $\ket{\vec{v}}$ such that ${1\over 2}(v_{2j-1}^{2}+v_{2j}^{2}) = {En_{j}\over N}$, $j=1,\ldots,m$ (this fact follows from a constrained likelihood maximization for independent Poisson random variables). Therefore, such a cost function can be replaced by a cost function defined by $H=\mathbb{I}-\ket{\vec{v}}\bra{\vec{v}}$ or any phase shifted image of this Hamiltonian. Since $\Vert \vec{v}\Vert = \Vert \vec{u} \Vert=\sqrt{2E}$, the analysis of BPL involves taking $\vec{b}=O_{+}\vec{v}^{T}$ in (\ref{eqn:simp}) and again using (\ref{eqn:prop1}).

Proposition \ref{prop:one} can also be used to analyze trainability of cost functions based on heterodyne or homodyne measurement outcomes. To optimize a homodyne measurement outcome, cost function (\ref{eqn:cf}) can be used with $H=\mathbb{I}-P$ with $P$ a projection onto an eigenvector of $RV$ for some $V\in O(2m)$ \footnote{Rigorously, $H$ should be considered as a limit of rank one projections because quadrature eigenvectors have undefined intensity.}, whereas to optimize a heterodyne measurement outcome $H$ can be taken as a projection onto a multimode coherent state with intensity not necessarily equal to the intensity of $\ket{\psi_{0}}=\ket{\vec{u}}$. In both cases, the cost function is minimized on a coherent state with a certain mean vector, but there is always a heterodyne measurement outcome that has the same cost function minimizer as for a homodyne measurement outcome, so we restrict our consideration to cost functions of the form $C(\vec{n},\theta)=1-\vert \langle \vec{u}\vert U(\theta)\vert \vec{n}\rangle \vert^{2}$ with $\vec{n}\in \mathbb{R}^{2m}$, and take $\Vert \vec{u}\Vert=\sqrt{2E_{0}}$, $\Vert \vec{n}\Vert=\sqrt{2E_{1}}$.
 If $E_{1}=0$, the cost function $C(\vec{n}=0,\theta)$ is independent of $\theta$, which implies BPL. Therefore, one expects that trainability of $C(\vec{n},\theta)$ depends on the scaling of both $E_{0}$ and $ E_{1}$. The same approach used in the proof of (\ref{eqn:prop1}) yields the interval $E_{O_{+},O_{-}}\left( \left(\partial_{\theta_{k}}C\right)^{2} \right) \in t[\xi_{\textup{min}},\xi_{\textup{max}}]$ with
\begin{align}
t&= {e^{-2(E_{1}+E_{0})}\Gamma(m)I_{m-1}(4\sqrt{E_{0}E_{1}})\over 2m(2\sqrt{E_{0}E_{1}})^{m-3}}.
\label{eqn:prop1homo}
\end{align}
It follows from (\ref{eqn:prop1homo}) and the BPL analysis in Appendix \ref{sec:app1} that the condition of sublinear scaling of $E_{0}$ is not sufficient alone to avoid BPL. An additional requirement on the measurement outcome range that guarantees trainability is that $E_{1}$ scales with $m$ between $O(1/\text{poly}(m))$ and $o(m)$, which results in a subexponential decrease to zero for (\ref{eqn:prop1homo}) according to the small argument asymptotics of the modified Bessel function.

Expression (\ref{eqn:prop1homo}) can also be used derive the dependence of BPL on noise strength for certain bosonic Gaussian noise channels, e.g., a linear optical RPQC interleaved with $L$ layers of quantum-limited attenuation channels. A translation invariant quantum-limited attenuator acts on a coherent state via $\mathcal{N}(\ket{\vec{u}}\bra{\vec{u}})= \ket{k\vec{u}}\bra{k\vec{u}}$ with $0<k<1$ \cite{holevoqubook}. Because $\mathcal{N}$ commutes with the action of any linear optical unitary channel, it suffices to consider $L$ applications of $\mathcal{N}$ at the end of the linear optical RPQC, leading to $E_{1}=k^{2L}E_{0}$ in (\ref{eqn:prop1homo}). For fixed noise strength $k$, and for sublinear scaling of $E_{0}$ (with $E_{0}$ not vanishing exponentially), BPL does not occur if the number of noise layers scales as $o(m)$. In the discrete variable case, it has been shown that BPL is induced from $L$ layers of a large set of local noisy channels if $L$ scales at least linearly with the number of qubits above a noise-dependent rate, or if L scales superlinearly for any noise strength \cite{nibp}.

\section{\label{sec:mfvqe}Mean field energy of quadratic Hamiltonians}

Minimization of the mean field energy of a positive Hamiltonian that is quadratic in the elements of $R$ \cite{bosequad} is important for variational estimation of quadrature fluctuations of CV quantum states. The cost function has the form (\ref{eqn:cf}) with $H=R\eta R^{T}$, $\eta$ a $2m\times 2m$ positive real matrix. Using the same RPQC structure as defined in Sec. \ref{sec:gen}, it follows that the derivative of the cost function with respect to $\theta_{k}$ then satisfies $\del_{\theta_{k}}C= \langle \vec{u}\vert U_{-}^{*}AU_{-}\vert \vec{u}\rangle$ where $A:=-i[R\epsilon_{k}R^{T},R\tilde{\eta}R^{T}]$ with $\tilde{\eta}:= O_{+}\eta O_{+}^{T}$. From the identity $[RMR^{T},RNR^{T}]=2iR\left(  M\Delta N  -N\Delta M \right)R^{T}$, where $M$ and $N$ are symmetric matrices and $\Delta$ is the standard symplectic form $\Delta=(i\sigma_{y})^{\oplus m}$ on $\mathbb{R}^{2m}$, the derivative simplifies to
\begin{equation}
\del_{\theta_{k}}C= \langle \vec{u}\vert O_{-}B_{k}O_{-}^{T}\vert \vec{u}\rangle
\end{equation}
where $B_{k}:= 2\epsilon_{k}\Delta \tilde{\eta} - 2\tilde{\eta} \Delta \epsilon_{k}=B_{k}^{T}$. Under the assumption that $[\epsilon_{k},\Delta]=0$, a constraint which we impose on the generators in this section, it follows that $\text{tr}B_{k}=0$. For example, a two-mode phase shifter on modes $i,j$ in the $k$-th layer given by $U_{k}(\theta_{k})=e^{-i{\theta_{k} \over 2}(q_{i}^{2}+p_{i}^{2}-q_{j}^{2}-p_{j}^{2})}$ satisfies the constraint (or any linear optical transformation of this unitary, e.g., a two-mode beamsplitter). Note that the structure of the RPQC in terms of $U_{\ell}(\theta_{\ell})$ and $W_{\ell}$ does not allow for $U(\theta)$ to involve parameterized squeezing followed by other layers that undo the squeezing, so the $k$-th layer can be written in terms of parameterized beam splitters and phase shifters.  The conditions $B_{k}=B_{k}^{T}$ and the assumption $[\epsilon_{k},\Delta]=0$ together imply that
\begin{equation}
\del_{\theta_{k}}C= \vec{u}O_{-}B_{k}O_{-}^{T}\vec{u}^{T}
\label{eqn:nnn}
\end{equation}
and, therefore, $E_{O_{-}}\left( \del_{\theta_{k}}C \right)=0$. The following proposition will allow to determine a sufficient condition on $\vec{u}$ such that the cost function $C$ does not exhibit BPL with respect to parameter $\theta_{k}$.

\begin{proposition}
Let $B_{k}$ be defined as above. Then
\begin{equation}
E_{O_{-}}\left( \left( \del_{\theta_{k}}C\right)^{2}\right)= {\Vert \vec{u}\Vert^{4}\over 2m(2m+2)}\left( \textup{tr}B_{k}^{2}+\Vert B_{k}\Vert_{F}^{2} \right)
\label{eqn:iii}
\end{equation}
where $\Vert \cdot \Vert_{F}$ is the Frobenius norm.
\label{prop:two}
\end{proposition}

\begin{proof}
The square of (\ref{eqn:nnn}) can be written
\begin{equation}
\Vert \vec{u} \Vert^{4}\sum_{j,s,j',s'}T_{1,j}T_{1,s}T_{1,j'}T_{1,s'}(B_{k})_{j,s}(B_{k})_{j',s'}
\end{equation}
where $T=O_{-}$ and the matrix elements are with respect to an orthonormal basis $\lbrace \vec{e}_{1}:={\vec{u}\over \Vert \vec{u} \Vert},\vec{e}_{2},\ldots, \vec{e}_{2m} \rbrace $.
 The expectation of (\ref{eqn:iii}) over $O_{-}$ involves four inequivalent contributions
\begin{align}
&{} 3g(m)\sum_{j}(B_{k})_{j,j}^{2}\nonumber \\
&+ g(m)\sum_{j\neq j'}(B_{k})_{j,j}(B_{k})_{j',j'} \nonumber \\
&+  g(m)\sum_{j\neq s}(B_{k})_{j,s}^{2} \nonumber \\
&+ g(m)\sum_{j\neq s}(B_{k})_{j,s}(B_{k})_{s,j}
\end{align}
where $g(m):=(2m(2m+2))^{-1}$. Combining the sums results in 
\begin{align}
E_{O_{-}}\left( \left( \del_{\theta_{k}}C\right)^{2}\right)&=\Vert \vec{u}\Vert^{4}g(m)\left( (\text{tr}B_{k})^{2}+ \textup{tr}B_{k}^{2}+\Vert B_{k}\Vert_{F}^{2} \right)
\end{align}
which reduces to (\ref{eqn:iii}) due to $\text{tr}B_{k}=0$.
\end{proof}

Subsequently taking the expectation over $O_{+}$ of (\ref{eqn:iii}) does not involve the input state and can only result in $1/\text{poly}(m)$ factors. Therefore, a corollary of Proposition \ref{prop:two} is that BPL is precluded for cost functions defined by positive quadratic Hamiltonians, unless the linear optical RPQC is positioned in the quantum optical neural network in such a way that the input state has exponentially attenuated intensity (with respect to the number of modes). 

\section{Discussion}
Our results demonstrating the trainability of generic linear optical submodules of photonic variational quantum algorithms for coherent input states are a first step toward the complete picture of trainability of random linear optical modules.
However, the restriction to coherent inputs imposes some limitations to generalizing the results of the present work to arbitrary linear optical QONN modules. For example, in our analysis of BPL of measurement outcomes for homodyne, heterodyne, and photon counting measurements in Sec. \ref{sec:gen}, we used the fact that the separability of coherent states allows to replace cost functions defined by $H=\mathbb{I}-\ket{\psi}\bra{\psi}$ for separable $\ket{\psi}$ by equivalent cost functions for which $\ket{\psi}$ is a coherent state. Because universal optical quantum computation requires non-linearities \cite{RevModPhys.84.621}, achieved via coherent non-Gaussian evolution such as Kerr or cross-Kerr interactions or via measurements such as photon counting \cite{PhysRevLett.88.097904}, the input states to linear optical modules in a near-term QONN may not be coherent states, or even Gaussian states. For instance, sufficient conditions for avoiding BPL for cost functions based on photon counting of generic linear optical orbits of Fock states (respectively, Gaussian states) would involve anticoncentration bounds for derivatives of permanents \cite{aaronson} (respectively, derivatives of Hafnians \cite{jex1,jex2}).  Our basic approach can be extended to analyze trainability for more general Gaussian submodules of QONN, e.g., with squeezed input states, linear optical RPQC, and Gaussian measurements, and for linear optical modules acting on non-Gaussian input states consisting of a superposition of a small number of coherent states. We also
expect that the methods of the present work can be applied to those classes of QONN for which the number of modes and intensity depend on
the layer of the QONN. For example, these features are present in QONN implementations of dissipative quantum neural
networks  \cite{osbornediss,dqnn}.

The analysis in Proposition \ref{prop:two} can be extended to the case of variational mean field energy estimation of positive quartic Hamiltonians by using higher moments of the Haar measure on the orthogonal group \cite{braun}. For few-mode bosonic systems, the structure of exact and high-quality variational ground states of certain quartic Hamiltonians is known \cite{ortiz,foerster,richardson,PhysRevA.100.022331,PhysRevA.94.042327}. These results for few mode bosonic systems can be used to inform the structure of RPQC $U(\theta)$ for efficiently compiling ground states of larger systems of interacting ultracold atoms in optical lattices.

\section*{Acknowledgments}
\pst
The author acknowledges support from the LDRD program at LANL. Los Alamos National Laboratory is managed by Triad National Security, LLC, for the National Nuclear Security Administration of the U.S. Department of Energy under Contract No. 89233218CNA000001.

\appendix
\section{\label{sec:app1}BPL analysis via asymptotics of $I_{m-1}(4E)$}

To determine whether the right hand side of (\ref{eqn:prop1}) satisfies the definition of BPL when a linear scaling $E\sim a(m-1)$ ($a>0$) is assumed for the intensity, it is appropriate to use the uniform asymptotics of the modified Bessel function (9.7.7 of \cite{as}), then argue according to the Chebyshev inequality (\ref{eqn:cheby}). The right hand side of (\ref{eqn:prop1}) is then given by
\begin{align}
{e^{-4E}\Gamma(m)I_{m-1}(4E)\over 2m(2E)^{m-3}}&\sim {e^{-4a(m-1)}(2a(m-1))^{2}\over 2m(2a)^{m-1}} {e^{-(m-1)}\sqrt{2\pi}\over \sqrt{m-1}}{e^{(m-1)\sqrt{16a^{2}+1}}\over \sqrt{2\pi (m-1)}(16a^{2}+1)^{1/4}}\left( {4a \over 1+\sqrt{16a^{2}+1} } \right)^{m-1}\nonumber \\
&= { (2a(m-1))^{2}e^{-(m-1)(4a+1-\sqrt{16a^{2}+1})}\over 2m(m-1)(16a^{2}+1)^{1/4} } \left( {2\over  1+\sqrt{16a^{2}+1} } \right)^{m-1}
\label{eqn:app1}
\end{align}
where in the first line we used the uniform asymptotics of the modified Bessel function $I_{m-1}(4a(m-1))$ and Stirling's approximation to $\Gamma(m)$ in the form ${\Gamma(m)\over (m-1)^{m-1}} \sim \sqrt{2\pi\over m-1}e^{-(m-1)}$ was used in the first line. In (\ref{eqn:app1}), $4a+1-\sqrt{16a^{2}+1} >0$ and $2/( 1+\sqrt{16a^{2}+1})<1$ for all $a>0$. Therefore, $E_{O_{+},O_{-}}\left( \left(\partial_{\theta_{k}}C\right)^{2} \right)$ goes to zero exponentially if $E\sim am$ for any $a>0$. This implies BPL for cost function $C$.

By contrast, one could consider sublinear scaling of the intensity via $E\sim a m^{r}$ with $r<1$ and $a>0$ (in particular, $E= o(m)$). For large $m$, the function $I_{m-1}(4am^{r})$ can be approximated using the small argument asymptotics (9.6.7 of \cite{as}).   The right hand side of (\ref{eqn:prop1}) is then given by
\begin{align}
{e^{-4am^{r}}\Gamma(m)I_{m-1}(4am^{r})\over 2m(2am^{r}))^{m-3}}&\sim {e^{-4am^{r}}(2am^{r})^{2}\over 2m}.
\end{align}
Since for any $c>0$, $\lim_{m\rightarrow \infty}e^{cm}e^{-4am^{r}} = \infty$ and since the left end of the interval in Proposition \ref{prop:one} is a lower bound on $E_{O_{+},O_{-}}\left( \left(\partial_{\theta_{k}}C\right)^{2} \right)$, one concludes that $E_{O_{+},O_{-}}\left( \left(\partial_{\theta_{k}}C\right)^{2} \right)$ does not decrease exponentially to zero and, therefore, $C$ does not exhibit BPL. The same small argument asymptotic formula for $I_{m-1}(4E)$ also implies that BPL exists if $E$ goes to zero exponentially with $m$.

\bibliographystyle{unsrt_style}
\bibliography{phasebib}

\end{document}